\documentclass[lettersize,journal]{IEEEtran}
\usepackage{amsmath,amsfonts}
\usepackage{algorithmic}
\usepackage{algorithm}
\usepackage{array}
\usepackage[caption=false,font=normalsize,labelfont=sf,textfont=sf]{subfig}
\usepackage{textcomp}
\usepackage{stfloats}
\usepackage{url}
\usepackage{verbatim}
\usepackage{bm}
\usepackage{graphicx}
\usepackage{cite}
\usepackage{hyperref}
\usepackage{multirow}

\usepackage{amsthm} 
\newtheorem{theorem}{Theorem}

\newtheorem{remark}{Remark}

\usepackage{xcolor}

\begin{document}

\title{GLASD: A Loss-Function-Agnostic Global Optimizer for Robust Correlation Estimation under Data Contamination and Heavy Tails}

\author{Priyam Das\\
~\IEEEmembership{Assistant Professor, Department of Biostatistics, Virginia Commonwealth University, Richmond, Virginia, USA}
\thanks{This paper was produced by the IEEE Publication Technology Group. They are in Piscataway, NJ.}
\thanks{Manuscript received June XX, 2025; revised August XX, 2025.}}

\markboth{Journal of \LaTeX\ Class Files,~Vol.~xx, No.~xx, August~2025}%
{Shell \MakeLowercase{\textit{et al.}}: A Sample Article Using IEEEtran.cls for IEEE Journals}


\maketitle

\begin{abstract}
Robust correlation estimation is essential in high-dimensional settings, particularly when data are contaminated by outliers or exhibit heavy-tailed behavior. Many robust loss functions of practical interest—such as those involving truncation or redescending M-estimators—lead to objective functions that are inherently non-convex and non-differentiable. Traditional methods typically focus on a single loss function tailored to a specific contamination model and develop custom algorithms tightly coupled with that loss, limiting generality and adaptability. We introduce GLASD (Global Adaptive Stochastic Descent), a general-purpose black-box optimization algorithm designed to operate over the manifold of positive definite correlation matrices. Unlike conventional solvers, GLASD requires no gradient information and imposes no assumptions of convexity or smoothness, making it ideally suited for optimizing a wide class of loss functions—including non-convex, non-differentiable, or discontinuous objectives. This flexibility allows GLASD to serve as a unified framework for robust estimation under arbitrary user-defined criteria. We demonstrate its effectiveness through extensive simulations involving contaminated and heavy-tailed distributions, as well as a real-data application to breast cancer proteomic network inference, where GLASD successfully identifies biologically plausible interactions despite the presence of outliers. The proposed method is scalable, constraint-aware, and available as open-source software at GitHub.
\end{abstract}


\section{Introduction}\label{sec:intro}
\IEEEPARstart{E}{stimation} of correlation matrices is a foundational task in multivariate analysis, with wide-ranging applications in genomics, finance, neuroimaging, and network biology~\cite{shafer2015genomics, bickel2008covariance, smith2011network, fan2016overview}. These domains frequently involve high-dimensional datasets that may be contaminated by outliers, corrupted by measurement error, or drawn from heavy-tailed distributions~\cite{lam2016robust}. Classical estimators—such as the sample covariance or correlation matrix—are highly sensitive to such deviations, underscoring the need for robust correlation estimation methods that remain reliable under data irregularities.

Robust correlation estimation typically relies on minimizing loss functions that reduce the influence of extreme observations. Examples include the Huber loss~\cite{huber1964robust, fan2018robust}, which smoothly transitions between squared and absolute loss; Tukey’s biweight loss~\cite{beaton1974,hardin2007robust}, a bounded loss function whose corresponding influence function assigns zero weight to extreme values, thereby suppressing their impact on estimation. While these loss functions enhance robustness, they also pose significant optimization challenges: the resulting objectives are often non-convex, non-differentiable, and must be minimized over the space of symmetric, positive definite correlation matrices with unit diagonals. Optimization over such constrained manifolds is an active area of research~\cite{absil2009optimization, boumal2023introduction}, yet standard gradient-based or projection-based solvers often struggle when the objective is nonsmooth or discontinuous.

In the context of model-based robust correlation matrix estimation, most existing approaches tailor their algorithms to a specific choice of loss function and structural assumption (e.g., \cite{fan2018robust}), leading to limited flexibility and generalizability. This makes it difficult for practitioners to experiment with alternative robustness criteria or adapt methods to new data domains. In this work, we introduce \textit{Global Adaptive Stochastic Descent} (GLASD), a general-purpose black-box optimization algorithm designed to robustly estimate correlation matrices under arbitrary user-specified loss functions. The method is fully gradient-free and combines adaptive stochastic descent~\cite{kerr2018optimization} with global exploration strategies inspired by simulated annealing~\cite{geman1984stochastic, hajek1988cooling}, making it well-suited for complex constrained spaces with rugged landscapes. Unlike conventional solvers, GLASD requires no convexity or smoothness assumptions, enabling it to optimize a wide range of objectives—including discontinuous and highly non-convex loss functions—on the manifold of positive definite correlation matrices.

Theoretical contributions of our work include a bijective parameterization of the correlation matrix space via a hyperspherical coordinate representation of its Cholesky factor. This formulation enables efficient, projection-free optimization while preserving symmetry, positive definiteness, and unit diagonals. We also establish global convergence guarantees under mild regularity conditions.

Beyond theory, we demonstrate the practical utility of GLASD through extensive simulation studies involving multiple contamination types and correlation structures. Finally, we apply GLASD to estimate robust correlation matrices for four major pathway-related proteins in breast cancer, using data from The Cancer Proteome Atlas (TCPA)~\cite{li2013tcpa}. These proteomic measurements exhibit outliers, making robust estimation essential. GLASD successfully uncovers biologically meaningful network structures, revealing insightful inter- and intra-pathway correlation patterns.

The rest of the paper is organized as follows. In Section~\ref{sec:problem}, we introduce the robust correlation matrix estimation problem under several loss functions. In Section~\ref{sec:GLASD}, we introduce the GLASD algorithm in general form. Section~\ref{sec:GLASD_theory} presents theoretical results that establish the global convergence guarantees of GLASD. In Section~\ref{sec:correlation_optimization}, we develop the Cholesky-hyperspherical parameterization of correlation matrices and adapt GLASD to this constrained manifold. Results from the benchmark evaluation and simulation experiments are given in Section~\ref{sec:GLASD_benchmark}, and ~\ref{sec:simulation_study}, respectively, followed by a real-data application in Section~\ref{sec:TCPA_analysis}. Finally we conclude with discussion in Section~\ref{sec:conclusion}.

\section{Robust Correlation Estimation via Loss-based Mahalanobis Distance Modeling}\label{sec:problem}

\noindent Estimating the correlation matrix of a multivariate distribution is central to many downstream tasks such as clustering, dimensionality reduction, graphical modeling, and portfolio optimization. Unlike the covariance matrix, the correlation matrix provides a scale-invariant summary of pairwise dependencies, which is especially useful in heterogeneous, high-dimensional settings. However, the sample correlation matrix is often unstable in high dimensions, particularly under limited sample sizes, outliers, or heavy-tailed noise. These challenges have motivated model-based estimation strategies, such as those based on Gaussian likelihood, where the correlation matrix is inferred through structured optimization.

A principled approach to correlation matrix estimation often begins with the assumption that observations $\{\boldsymbol{x}_i\}_{i=1}^n \subset \mathbb{R}^p$ arise from a multivariate normal distribution with zero mean and unknown correlation matrix $\boldsymbol{C}$. In this setting, we assume that the observed data have been standardized to have mean zero and unit variance, so that the matrix $\boldsymbol{C}$ captures the correlation structure. Under this model, the negative log-likelihood function (up to an additive constant) takes the following form. 
\begin{equation}
    \ell_{\mathrm{Gauss}}(\boldsymbol{C}) = \frac{n}{2} \log\det(\boldsymbol{C}) + \frac{1}{2} \sum_{i=1}^n \boldsymbol{x}_i^\top \boldsymbol{C}^{-1} \boldsymbol{x}_i,
    \label{eq:gaussian_likelihood}
\end{equation}
where the constraint $\text{diag}(\boldsymbol{C}) = \boldsymbol{1}$ ensures unit variances, and $\boldsymbol{C} \succ 0$ enforces positive definiteness. The second term, based on squared Mahalanobis distances $d_i^2 = \boldsymbol{x}_i^\top \boldsymbol{C}^{-1} \boldsymbol{x}_i$, measures the discrepancy of each observation from the origin relative to the geometry induced by $\boldsymbol{C}$.

While optimal under Gaussianity, this estimator is highly sensitive to outliers and heavy-tailed behavior—conditions frequently encountered in genomics, finance, and biomedical network modeling~\cite{fan2018robust, lam2016robust}. Even a single contaminated observation can yield inflated Mahalanobis distances, destabilizing the entire correlation estimate. Therefore, to improve robustness, several alternatives have been proposed that replace the squared distance term with bounded or redescending loss functions. One such method is the \emph{Huber loss}~\cite{fan2018robust}, which behaves quadratically for small residuals and linearly for large ones:
\begin{equation}
\rho_{\delta}^{\mathrm{Huber}}(d^2) =
\begin{cases}
    d^2, & \text{if } d^2 \leq \delta, \\
    2\delta^{1/2} d - \delta, & \text{if } d^2 > \delta.
\end{cases}
\end{equation}
This yields the robustified log-likelihood (up to an additive constant):
\begin{equation}
    \ell_{\mathrm{Huber}}(\boldsymbol{C}) = \frac{n}{2} \log\det(\boldsymbol{C}) + \frac{1}{2} \sum_{i=1}^n \rho_{\delta}^{\mathrm{Huber}}\left( \boldsymbol{x}_i^\top \boldsymbol{C}^{-1} \boldsymbol{x}_i \right),
    \label{eq:huber_likelihood}
\end{equation}
where $\delta > 0$ is a tuning parameter controlling the threshold between quadratic and linear behavior.

To address more severe contamination, we also consider the \emph{truncated loss}, which entirely caps large residuals:
\begin{equation}
    \ell_{\mathrm{Trunc}}(\boldsymbol{C}) = \frac{n}{2} \log\det(\boldsymbol{C}) + \frac{1}{2} \sum_{i=1}^n \min\left\{ \boldsymbol{x}_i^\top \boldsymbol{C}^{-1} \boldsymbol{x}_i, \, \tau \right\},
    \label{eq:truncated_loss}
\end{equation}
where $\tau > 0$ is a fixed truncation threshold. This formulation discards the influence of observations whose Mahalanobis distance exceeds $\tau$, offering a stricter form of robustness.

Finally, we explore the \emph{Tukey’s biweight loss}~\cite{beaton1974}, a redescending function that limits large deviations and smoothly downweights moderate outliers. Its form is:
\begin{equation}
\rho_{\tau}^{\mathrm{Tukey}}(d^2) =
\begin{cases}
    \frac{\tau^2}{6} \left[1 - \left(1 - \frac{d^2}{\tau^2}\right)^3 \right], & \text{if } d^2 \leq \tau^2, \\
    \frac{\tau^2}{6}, & \text{if } d^2 > \tau^2.
\end{cases}
\end{equation}
The resulting objective becomes:
\begin{equation}
    \ell_{\mathrm{Tukey}}(\boldsymbol{C}) = \frac{n}{2} \log\det(\boldsymbol{C}) + \frac{1}{2} \sum_{i=1}^n \rho_{\tau}^{\mathrm{Tukey}}\left( \boldsymbol{x}_i^\top \boldsymbol{C}^{-1} \boldsymbol{x}_i \right).
    \label{eq:tukey_likelihood}
\end{equation}

Each of these loss functions provides a different trade-off between statistical efficiency and robustness. However, they all result in non-convex, often non-smooth objective functions defined over the manifold of correlation matrices (i.e., symmetric, positive definite matrices with unit diagonal). These properties render classical optimization approaches—such as gradient-based or projection-based schemes—ineffective or computationally prohibitive.

To address this, we introduce \emph{GLASD (Global Adaptive Stochastic Descent)}, a general-purpose black-box optimization algorithm capable of handling arbitrary user-defined loss functions, regardless of smoothness or convexity. GLASD navigates the constrained space of correlation matrices using adaptive stochastic exploration, enabling robust estimation in high-dimensional, contaminated settings without requiring gradient information or problem-specific tuning.

\section{Global Adaptive Stochastic Descent (GLASD)}\label{sec:GLASD}

\noindent Robust correlation estimation involves optimizing non-convex, potentially non-smooth objectives over the manifold of symmetric positive definite matrices with unit diagonal constraints. As discussed in the previous section, standard solvers—whether gradient-based or manifold-constrained—are often ineffective in this setting due to discontinuities, flat regions, or the difficulty of enforcing structural constraints during optimization.

To overcome these limitations, we introduce \emph{Global Adaptive Stochastic Descent (GLASD)}, a randomized, black-box optimization algorithm tailored for solving non-smooth and non-convex problems over compact hyperrectangular domains. Although our eventual goal is to minimize robust loss functions over the space of correlation matrices, we first formulate GLASD in a general hyperrectangular setting. This abstraction is justified by the fact that correlation matrices can be bijectively parameterized via a constrained Euclidean representation, as discussed later in Section~\ref{sec:correlation_optimization}.

Let \( f: D \subset \mathbb{R}^n \to \mathbb{R} \) denote the objective function, where \( D = \prod_{i=1}^n [a_i, b_i] \) is a compact hyperrectangle obtained through such transformation. GLASD maintains a current iterate \( x_t \in D \) and associates with each coordinate-direction \( j \in \{1, \dots, 2n\} \) a step size \( s_j^{(t)} > 0 \) and a directional selection probability \( p_j^{(t)} > 0 \), normalized so that \( \sum_j p_j^{(t)} = 1 \). The directions \( j = 2i - 1 \) and \( j = 2i \) represent the positive and negative perturbations along coordinate \( i \), respectively. Suppose, $s_{\text{inc}} > 1$ and $s_{\text{dec}} > 1$ denote the multiplicative factors used to increase or decrease the coordinate-wise step sizes, while $p_{\text{inc}} > 1$ and $p_{\text{dec}} > 1$ control the corresponding updates to directional selection probabilities. At each iteration \( t \), GLASD proceeds in one of two modes:

\vspace{0.2cm}
\noindent \textbf{Greedy Descent Mode}: With probability \( 1 - \frac{1}{m} \), a direction \( j \) is sampled according to \( p^{(t)} \). Let \( i = \lceil j/2 \rceil \) and set the sign to \( +1 \) if \( j \) is odd, and \( -1 \) otherwise. A step of magnitude \( s_j^{(t)} \) is proposed in the corresponding direction and clipped to ensure feasibility:
\[
\delta[i] = \begin{cases}
    \min(s_j^{(t)}, \tfrac{b_i - x_t[i]}{2}) & \text{if } \text{sign} = +1, \\
    -\min(s_j^{(t)}, \tfrac{x_t[i] - a_i}{2}) & \text{if } \text{sign} = -1,
\end{cases}
\]
with \( \delta[k] = 0 \) for all \( k \ne i \). If the proposed point satisfies \( f(x_t + \delta) < f(x_t) \), the update is accepted:
\[
x_{t+1} \gets x_t + \delta, \quad s_j^{(t+1)} \gets s_j^{(t)} \cdot s_{\text{inc}}, \quad p_j^{(t+1)} \gets p_j^{(t)} \cdot p_{\text{inc}},
\]
followed by normalization of the updated probability vector. Otherwise, the step is rejected and parameters are decayed:
\[
s_j^{(t+1)} \gets s_j^{(t)} / s_{\text{dec}}, \quad p_j^{(t+1)} \gets p_j^{(t)} / p_{\text{dec}},
\]
again, followed by normalization of the updated probability vector.

\vspace{0.2cm}
\noindent \textbf{Forced Exploration Mode}: With probability \( \frac{1}{m} \), the algorithm performs a random exploration step. A coordinate \( i \in \{1, \dots, n\} \) is chosen uniformly, a sign is sampled from \( \{+1, -1\} \), and a step magnitude \( s \sim \mathcal{U}(0, r) \) is drawn, where \( r > 0 \) is a fixed upper bound controlling the maximum exploration radius. Let \( t \) denote the current iteration number. The step is again clipped:
\[
\delta[i] = \begin{cases}
    \min(s, \tfrac{b_i - x_t[i]}{2}) & \text{if } \text{sign} = +1, \\
    -\min(s, \tfrac{x_t[i] - a_i}{2}) & \text{if } \text{sign} = -1,
\end{cases}
\]
with \( \delta[k] = 0 \) for all \( k \ne i \). The move is accepted if the objective function value decreases at the proposed point. Otherwise, if the objective value increases, the move is accepted with probability
\[
q_t = \min\left(1, \frac{m c}{\log(1 + t)}\right),
\]
where \( c > 0 \) controls the exploration temperature, analogous to the role of a cooling schedule in simulated annealing~\cite{hajek1988cooling}. This particular form of the acceptance probability facilitates global convergence guarantees, as shown later in Section~\ref{sec:GLASD_theory}. The forced exploration step allows GLASD to escape local optima in case it becomes trapped in one.

GLASD maintains a running minimum \( f^{\text{best}}_t = \min \{f(x_0), \dots, f(x_t)\} \), and terminates if either a maximum number of iterations \( T \) is reached or if no significant improvement (\( > \epsilon \)) is observed over the last \( M \).

GLASD is projection-free and zeroth-order: all updates are feasible by construction, and only function evaluations are required. GLASD is inspired by \emph{Adaptive Stochastic Descent (ASD)}~\cite{kerr2018optimization}, but differs in three critical aspects: (i) it is designed for compact domains unlike ASD which is on unconstrained domain, (ii) it introduces exploration via probabilistic uphill moves, and (iii) it comes with theoretical convergence guarantees under mild assumptions. The GLASD algorithm on hyper-rectangle is detailed below in Algorithm \ref{alg:glasd}.

\begin{algorithm}[H]
\caption{Global Adaptive Stochastic Descent (GLASD)}
\begin{algorithmic}
\STATE \textsc{GLASD} optimizes $f(x)$ on $x\in D = \prod_{i=1}^n [a_i, b_i]$
\STATE \textbf{User input:} Initial point $x_0 \in D$; $\{a_i,b_i\}$ for $i = 1, \dots, n$ s.t. $a_i < b_i$.
\STATE \textbf{Optional input (else use default):} Step sizes $s_j > 0$ and probabilities $p_j > 0$ for $j = 1, \dots, 2n$ with $\sum_j p_j = 1$; $s_{\text{inc}}, s_{\text{dec}} > 1$; $p_{\text{inc}}, p_{\text{dec}} > 1$; $m \in \mathbb{N}$; $r > 0$; $c > 0$; maximum iterations $T$; stagnation window $M$; tolerance $\epsilon$.
\vspace*{0.01em}
\STATE \textbf{Algorithm:} Initialize $x \gets x_0$, $E \gets f(x)$, $E^{\text{best}} \gets E$, buffer $B \gets [E^{\text{best}}]$
\FOR{$k = 1$ to $T$}
\STATE \hspace{0.5cm} $q_t \gets min(1, mc/\log(1 + k))$
\STATE \hspace{0.5cm} \textbf{if} Bernoulli$(1 - 1/m)$
\STATE \hspace{1cm} \textit{explore} $\gets$ \textbf{false}
\STATE \hspace{1cm} Sample $j \in \{1,\dots,2n\}$ by $p$; $i \gets \lceil j/2 \rceil$
\STATE \hspace{1cm} $\text{sign} \gets +1$ if $j$ odd; $-1$ otherwise
\STATE \hspace{1cm} \textbf{if} $\text{sign} = +1$ \textbf{then} $\delta[i] \gets min(s_j, (b_i - x[i])/2)$
\STATE \hspace{1cm} \textbf{else} $\delta[i] \gets -min(s_j, (x[i] - a_i)/2)$
\STATE \hspace{1cm} Set $\delta[\ell] \gets 0$ for all $\ell \ne i$
\STATE \hspace{0.5cm} \textbf{else}
\STATE \hspace{1cm} \textit{explore} $\gets$ \textbf{true}
\STATE \hspace{1cm} Sample $i \in \{1,\dots,n\}$, $\text{sign} \in \{\pm1\}$, $s \sim \mathcal{U}(0,r)$
\STATE \hspace{1cm} \textbf{if} $\text{sign} = +1$ \textbf{then} $\delta[i] \gets min(s, (b_i - x[i])/2)$
\STATE \hspace{1cm} \textbf{else} $\delta[i] \gets -min(s, (x[i] - a_i)/2)$
\STATE \hspace{1cm} Set $\delta[\ell] \gets 0$ for all $\ell \ne i$
\STATE \hspace{0.5cm} $E_{\text{new}} \gets f(x + \delta)$
\STATE \hspace{0.5cm} \textbf{if} $E_{\text{new}} < E$
\STATE \hspace{1cm} $x \gets x + \delta$, $E \gets E_{\text{new}}$
\STATE \hspace{1cm} \textbf{if} \textit{explore} = \textbf{false}
\STATE \hspace{1.5cm} $s_j \gets s_j \cdot s_{\text{inc}}$, $p_j \gets p_j \cdot p_{\text{inc}}$, normalize $p$
\STATE \hspace{0.5cm} \textbf{else if} \textit{explore} = \textbf{true} \textbf{and} accept with prob. $q_t$
\STATE \hspace{1cm} $x \gets x + \delta$, $E \gets E_{\text{new}}$
\STATE \hspace{0.5cm} \textbf{else if} \textit{explore} = \textbf{false}
\STATE \hspace{1cm} $s_j \gets s_j / s_{\text{dec}}$, $p_j \gets p_j / p_{\text{dec}}$, normalize $p$
\STATE \hspace{0.5cm} $E^{\text{best}} \gets min(E^{\text{best}}, E)$; append to $B$
\STATE \hspace{0.5cm} $x^{\text{best}} \gets x^{E^{\text{best}}}$
\STATE \hspace{0.5cm} \textbf{if} $k \geq M$ \textbf{and} $B[k-M] - B[k] < \epsilon$ \textbf{then break}
\ENDFOR
\STATE $\textbf{return } x^{\text{best}}$\;\;\;\;\\
\end{algorithmic}
\label{alg:glasd}
\end{algorithm}
To offer an early illustration of GLASD’s convergence behavior, we present Figure~\ref{fig:glasd_benchmark}, which compares GLASD against several classical optimization strategies on five canonical test functions: Ackley, Griewank, Rastrigin, Rosenbrock, and Sumsquares~\cite{surjanovic2013virtual, molga2005test}. These benchmark functions, widely used for evaluating optimization algorithms, are all non-convex except for the Sumsquares function, which is convex. All functions are implemented in \textsc{MATLAB}, with consistent domain bounds and problem dimensionality settings as described in \cite{surjanovic2013virtual}.

The competing algorithms include three global optimizers—Simulated Annealing (\texttt{simulannealbnd})~\cite{kirkpatrick1983optimization}, Genetic Algorithm (\texttt{ga})~\cite{holland1992adaptation}, and Particle Swarm Optimization (\texttt{particleswarm})~\cite{kennedy1995particle}—and three local solvers from the \texttt{fmincon} family: Interior-Point, Sequential Quadratic Programming (SQP), and Active-Set~\cite{byrd2000trust, nocedal2006numerical}. We also include Pattern Search (\texttt{patternsearch})~\cite{torczon1997convergence} as a derivative-free local search baseline. The convergence profiles plot the log-scaled objective values versus the number of function evaluations. GLASD demonstrates strong and often superior convergence across all five functions.

These experiments serve as a preliminary benchmark for GLASD on general box-constrained domains. A more comprehensive performance comparison is presented in Section~\ref{sec:correlation_optimization}, where we evaluate GLASD and competing methods on benchmark functions defined over the space of correlation matrices.

\begin{figure*}[t]
  \centering
  \includegraphics[width=0.85\textwidth]{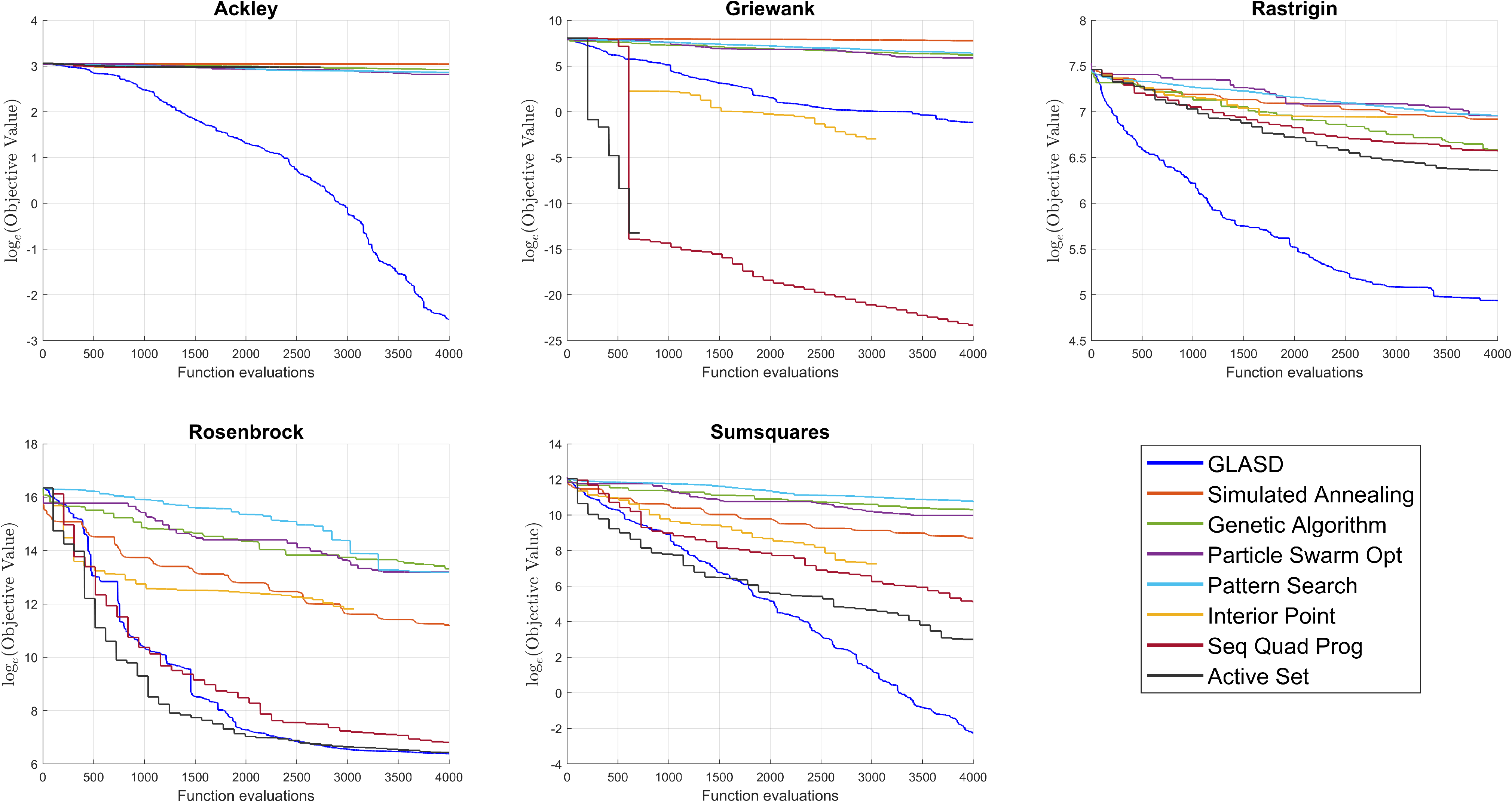}
   \caption{Convergence profiles of GLASD and competing optimization methods on benchmark functions. Log-scaled objective values are plotted against the number of function evaluations.}
  \label{fig:glasd_benchmark}
\end{figure*}

\section{Theoritical Properties}\label{sec:GLASD_theory}
\noindent The Global Adaptive Stochastic Descent (GLASD) algorithm combines greedy local search with probabilistic exploration to ensure robustness in optimizing non-convex functions. We first demonstrate that GLASD converges almost surely to the set of global minimizers under minimal assumptions—namely, continuity of the objective function and boundedness of the domain (Theorem \ref{thm:GLASD_glob_conv}). This result relies on a forced exploration mechanism, whose acceptance probability follows a logarithmic cooling schedule inspired by classical simulated annealing theory. We then turn our attention to the deterministic variant, Adaptive Stochastic Descent (ASD), which omits exploration and relies solely on monotonic improvement. While ASD lacks global convergence guarantees in non-convex settings, it admits a sharper characterization under convexity: we derive a linear convergence rate under strong convexity (Theorem \ref{thm:ASD_conv_rate}) and almost sure convergence under general convexity (Theorem \ref{thm:ASD_as_conv}). Together, these results provide a comprehensive theoretical foundation for understanding the convergence behavior of both GLASD and its exploration-free counterpart.

\begin{theorem}[Global convergence of GLASD]
Let $f: D \to \mathbb{R}$ be a continuous function defined on a compact hyperrectangle domain $D = \prod_{i=1}^n [a_i, b_i] \subset \mathbb{R}^n$. Suppose the sequence $\{x_t\}$ is generated by the GLASD algorithm, with the following properties:

(i) With probability $1 - 1/m$, a direction $j \in \{1, \dots, 2n\}$ is sampled from a fixed probability distribution $p$ with strictly positive entries, and a step $\delta$ is constructed along the associated coordinate-direction with step size $s_j$, clipped to ensure feasibility, i.e., $x_t + \delta \in D$;

(ii) With probability $1/m$, a forced exploration step is performed, in which a coordinate $i \in \{1, \dots, n\}$, a sign $\in \{+1, -1\}$, and a magnitude $s \sim \mathcal{U}(0, r)$ are drawn independently, and the resulting step is also clipped to satisfy $x_t + \delta \in D$;

(iii) Any forced exploration move that does not decrease the objective is accepted with probability
\[
q_t \geq \frac{m c}{\log(1 + t)} \quad \text{for some constant } c > 0.
\]

Then the sequence $\{x_t\}$ converges almost surely to the set of global minimizers of $f$; that is,
\[
\lim_{t \to \infty} \mathbb{P}(x_t \in \mathcal{X}^*) = 1, \quad \text{where } \mathcal{X}^* = \arg\min_{x \in D} f(x).
\]
\label{thm:GLASD_glob_conv}
\end{theorem}

\begin{proof}
Since the domain $D = \prod_{i=1}^n [a_i, b_i]$ is compact and the objective function $f$ is continuous, the minimum value $f^* = \min_{x \in D} f(x)$ is attained at some point $x^* \in D$, and the set of global minimizers $\mathcal{X}^* = \arg\min_{x \in D} f(x)$ is nonempty and closed.

The GLASD algorithm defines a time-inhomogeneous Markov process over $D$ by generating proposals along coordinate directions and accepting them according to a mixture of greedy and exploratory rules. At each iteration $t$, with probability $1/m$, the algorithm enters a forced exploration mode, where it selects a coordinate uniformly at random, chooses a direction uniformly from $\{+1, -1\}$, and samples a step size from a uniform distribution on $[0, r]$. The proposed step is then clipped to ensure that the updated iterate $x_t + \delta$ remains within $D$. Because the domain is a hyperrectangle and proposals affect only one coordinate at a time, this construction guarantees that the proposal kernel has support covering every axis-aligned neighborhood within $D$.

Hence, for any open neighborhood $U \subset D$ containing a global minimizer, there exists a nonzero probability of generating a proposal within $U$ during a forced exploration step. This ensures that the Markov process induced by GLASD is $\psi$-irreducible with respect to the Lebesgue measure on $D$.

To ensure convergence to a global minimizer, the algorithm must allow sufficiently frequent acceptance of exploratory steps. The acceptance probability for non-improving forced exploration moves is defined to satisfy $q_t \geq mc/\log(1 + t)$ for some $c > 0$, where $q_t$ decays sub-logarithmically in $t$. Consequently, the overall probability of making an uphill move at time $t$ (i.e., iteration) is at least $c/\log(1 + t)$. This cooling schedule aligns with the theoretical framework of simulated annealing, wherein logarithmic schedules are known to be both necessary and sufficient for convergence to global minima \cite{geman1984stochastic,hajek1988cooling}.

Specifically, \cite{hajek1988cooling} establishes that under mild regularity conditions, a time-inhomogeneous Markov chain with irreducible proposals and logarithmic acceptance probabilities converges almost surely to the set of global minimizers. Applying this result to GLASD, it follows that
\[
\lim_{t \to \infty} \mathbb{P}(x_t \in \mathcal{X}^*) = 1,
\]
which completes the proof.
\end{proof}


\begin{theorem}[Linear Convergence Rate of ASD]
Let $f : \mathbb{R}^n \to \mathbb{R}$ be a $\mu$-strongly convex and $L$-smooth function. Suppose the Adaptive Stochastic Descent (ASD) algorithm is run with the following conditions:
\begin{itemize}
    \item[(a)] Each iteration selects a direction $j \in \{1, \dots, 2n\}$ with probability $p_j \geq \delta > 0$;
    \item[(b)] Each direction $j$ has an associated step size $s_j \geq \underline{s} > 0$;
    \item[(c)] Steps are accepted only if they strictly reduce the function value;
    \item[(d)] Exact gradients are available.
\end{itemize}
Then, after $T$ successful (i.e., accepted) steps, the expected function value satisfies
\[
\mathbb{E}[f(x_T) - f^*] \leq \left(1 - \frac{\mu \cdot \underline{s}^2 \cdot \delta}{L} \right)^T (f(x_0) - f^*),
\]
where $f^* = \min_{x \in \mathbb{R}^n} f(x)$. That is, ASD achieves a linear convergence rate in expectation.
\label{thm:ASD_conv_rate}
\end{theorem}

\begin{proof}
Since $f$ is $L$-smooth, for any unit vector $v$ and step size $\alpha > 0$, the descent lemma gives
\[
f(x + \alpha v) \leq f(x) + \alpha \nabla f(x)^T v + \frac{L}{2} \alpha^2.
\]
Now consider a direction $j \in \{1, \dots, 2n\}$ corresponding to unit vector $v_j$ and step size $s_j$. A proposed move $x' = x + s_j v_j$ is accepted only if it strictly decreases $f$, and smoothness ensures that such a step achieves at least
\[
f(x + s_j v_j) \leq f(x) - \frac{s_j^2}{2L} (\nabla f(x)^T v_j)^2.
\]

Taking expectation over the sampled direction $j$, and using $p_j \geq \delta$, we have
\begin{align*}
\mathbb{E}_{j}\left[ (\nabla f(x)^T v_j)^2 \right] &= \sum_{j=1}^{2n} p_j (\nabla f(x)^T v_j)^2\\
& \geq \delta \sum_{j=1}^{2n} (\nabla f(x)^T v_j)^2 = \delta \|\nabla f(x)\|^2.
\end{align*}
Thus,
\[
\mathbb{E}[f(x_{t+1}) \mid x_t] \leq f(x_t) - \frac{\underline{s}^2 \delta}{2L} \|\nabla f(x_t)\|^2.
\]

Now apply strong convexity of $f$:
\[
\|\nabla f(x_t)\|^2 \geq 2\mu (f(x_t) - f^*).
\]

Combining the two bounds gives
\[
\mathbb{E}[f(x_{t+1}) - f^*] \leq \left(1 - \frac{\mu \cdot \underline{s}^2 \cdot \delta}{L} \right) (f(x_t) - f^*).
\]

Unrolling this recurrence for $T$ accepted iterations yields
\[
\mathbb{E}[f(x_T) - f^*] \leq \left(1 - \frac{\mu \cdot \underline{s}^2 \cdot \delta}{L} \right)^T (f(x_0) - f^*).
\]
\end{proof}

\begin{theorem}[Almost Sure Convergence of ASD]
Let $f : \mathbb{R}^n \to \mathbb{R}$ be a convex and continuously differentiable function. Suppose the Adaptive Stochastic Descent (ASD) algorithm is run with the following conditions:
\begin{itemize}
    \item[(a)] At each iteration $t$, a direction $j \in \{1, \dots, 2n\}$ is selected independently with probability $p_j \geq \delta > 0$;
    \item[(b)] Each direction $j$ has an associated step size $s_j^{(t)} > 0$, and for all $t$, $\underline{s} \leq s_j^{(t)} \leq \overline{s}$;
    \item[(c)] A proposed step $x_{t+1} = x_t + s_j^{(t)} v_j$ is accepted only if $f(x_{t+1}) < f(x_t)$;
    \item[(d)] The function $f$ is bounded below.
\end{itemize}
Then the sequence $\{f(x_t)\}_{t \geq 0}$ converges almost surely, and every limit point $x_\infty$ of the iterate sequence satisfies:
\[
f(x_\infty) = \inf_{t \geq 0} f(x_t).
\]
Moreover, if the set of minimizers of $f$ is nonempty, then:
\[
\lim_{t \to \infty} \text{dist}(x_t, \arg\min f) = 0 \quad \text{almost surely}.
\]
\label{thm:ASD_as_conv}
\end{theorem}

\begin{proof}
Let $\{x_t\}_{t \geq 0}$ denote the sequence of iterates generated by ASD. Because steps are accepted only if they strictly reduce the function value, the sequence $\{f(x_t)\}$ is strictly decreasing unless the algorithm rejects all proposed steps at iteration $t$.

Moreover, since $f$ is bounded below, the sequence $\{f(x_t)\}$ is bounded and monotonic, hence convergent:
\[
f(x_t) \downarrow f_\infty \quad \text{as } t \to \infty.
\]

Let $x_\infty$ denote a limit point of the sequence $\{x_t\}$ (guaranteed to exist since the step sizes are bounded and hence the sequence is tight). We want to show that $x_\infty$ achieves the minimal value among all iterates. Suppose, for contradiction, that there exists some $\epsilon > 0$ such that:
\[
f(x_\infty) > f_\infty + \epsilon.
\]

Since $f(x_t) \to f_\infty$, there exists a time $T$ such that for all $t \geq T$, we have:
\[
f(x_t) < f(x_\infty) - \epsilon/2.
\]

But by continuity of $f$ and the fact that $x_t \to x_\infty$ along some subsequence, we must also have:
\[
\lim_{t \to \infty} f(x_t) = f(x_\infty),
\]
contradicting the previous inequality. Hence, all limit points $x_\infty$ satisfy:
\[
f(x_\infty) = \inf_{t \geq 0} f(x_t).
\]

Next, suppose the set of minimizers $\mathcal{X}^* = \arg\min f$ is nonempty. Since $f$ is convex and continuous, this set is closed and convex. If $x_\infty \notin \mathcal{X}^*$, then $f(x_\infty) > f^*$, where $f^* = \min_x f(x)$. But from earlier we have $f(x_\infty) = \inf_t f(x_t) \to f^*$, so this contradicts the assumption that $x_\infty$ is not a minimizer.

Therefore, all limit points lie in $\mathcal{X}^*$, and because the sequence makes bounded-size moves, we conclude:
\[
\lim_{t \to \infty} \text{dist}(x_t, \mathcal{X}^*) = 0 \quad \text{almost surely}.
\]
\end{proof}

\begin{remark}
Theorem \ref{thm:ASD_conv_rate} and \ref{thm:ASD_as_conv} characterize the convergence behavior of the Adaptive Stochastic Descent (ASD) algorithm, which omits forced exploration. Unlike GLASD, ASD only accepts strictly improving updates, leading to a monotonic decrease in the objective value. This structure enables classical convergence analysis: a linear convergence rate under strong convexity, and almost sure convergence under general convexity. In contrast, GLASD permits occasional acceptance of non-improving exploratory steps, and thus its theoretical guarantees must rely on probabilistic convergence arguments (as in Theorem~\ref{thm:GLASD_glob_conv}) rather than deterministic descent.
\end{remark}

\section{GLASD over the Correlation Matrix Space}\label{sec:correlation_optimization}
\noindent This section outlines a transformation that maps the space of full-rank correlation matrices to a domain with simpler geometry, enabling convenient exploration of the search space within a Euclidean domain. A full-rank correlation matrix \( \boldsymbol{C} \in \mathbb{R}^{M \times M} \) (here $M$ denotes the dimension of the matrix, different from what $M$ denotes in the context of algorithm) is a symmetric positive definite matrix with unit diagonal entries. Such a matrix can be uniquely represented via its Cholesky decomposition as \( \boldsymbol{C} = \boldsymbol{L} \boldsymbol{L}^\top \), where \( \boldsymbol{L} \) is a lower triangular matrix with strictly positive diagonal entries. This decomposition provides a one-to-one correspondence between correlation matrices and a structured class of lower triangular matrices. The unit diagonal constraint of \( \boldsymbol{C} \) imposes implicit constraints on the elements of \( \boldsymbol{L} \), but the decomposition remains unique when the diagonal entries of \( \boldsymbol{L} \) are constrained to be positive. This parameterization is especially valuable in optimization, as it allows re-expressing the complex geometry of the correlation matrix space in terms of unconstrained or simply constrained elements of \( \boldsymbol{L} \), facilitating numerical search and efficient algorithmic design. To this end, Theorem~\ref{thm:bijection} constructs a bijective mapping from the space of full-rank correlation matrices to an open hyperrectangle in Euclidean space.

\begin{theorem}[Bijection Between Angular Coordinates and Correlation Matrices]
Let $\bm{L} = (l_{mg})$ be a lower triangular matrix such that $\bm{C} = \bm{L} \bm{L}^\top$ is a full-rank correlation matrix, i.e., a symmetric positive definite matrix with unit diagonal.

Define:
\begin{itemize}
    \item $l_{11} = 1$,
    \item for $m = 2$, the vector $\bm{l}_2 = (l_{21}, l_{22})$ lies on the unit semicircle $\mathbb{S}_+^1 = \{(x_1,x_2)\in\mathbb{R}^2: x_1^2 + x_2^2 = 1,\ x_2 > 0\}$, and is parameterized by a single angle $\omega_{21} \in (-\pi/2, \pi/2)$ via $l_{21} = \sin\omega_{21},\ l_{22} = \cos\omega_{21}$,
    \item for $m \geq 3$, the vector $\bm{l}_m = (l_{m1}, \dots, l_{mm})$ lies in $\mathbb{S}_+^{m-1} := \{x \in \mathbb{R}^m : \|x\| = 1,\ x_m > 0\}$ and is parameterized by angles $\Omega_m = (\omega_{m1}, \dots, \omega_{m,m-1}) \in \Theta_m := (0, \pi/2) \times (0, \pi)^{m-2} \times (0, 2\pi)$, via
\end{itemize}
\begin{align*}
&l_{m1} = \left( \prod_{i=1}^{m-2} \sin\omega_{mi} \right) \sin\omega_{m,m-1}, \\
&l_{m2} = \left( \prod_{i=1}^{m-2} \sin\omega_{mi} \right) \cos\omega_{m,m-1}, \\
&l_{m3} = \left( \prod_{i=1}^{m-3} \sin\omega_{mi} \right) \cos\omega_{m,m-2}, \\
&\;\;\;\;\;\;\;\;\vdots \\
&l_{m,m-1} = \sin\omega_{m1} \cos\omega_{m2}, \\
&l_{mm} = \cos\omega_{m1}.
\end{align*}

Then, the full collection of angular coordinates $\{\omega_{21}, \Omega_3, \ldots, \Omega_M\}$ uniquely parameterizes the Cholesky factor $\bm{L}$, and hence uniquely determines the correlation matrix $\bm{C} = \bm{L}\bm{L}^\top$. Conversely, any full-rank correlation matrix $\bm{C}$ can be mapped to a unique set of such angular coordinates.
\label{thm:bijection}
\end{theorem}

\begin{proof}
We prove both directions of the bijection. \\
\emph{Part 1, Angular coordinates define a valid correlation matrix:} Let $\{\omega_{21}, \Omega_3, \ldots, \Omega_M\}$ be a valid set of angles in the prescribed domains. Define $\bm{L}$ as:
\begin{itemize}
    \item $l_{11} = 1$,
    \item $l_{21} = \sin\omega_{21}$ and $l_{22} = \cos\omega_{21}$,
    \item for $m \geq 3$, construct $\bm{l}_m$ using the recursive formula above.
\end{itemize}
Each such $\bm{l}_m$ lies on the unit sphere with positive last coordinate, i.e., $\|\bm{l}_m\| = 1$ and $l_{mm} > 0$. Thus,
\begin{itemize}
    \item $\bm{L}$ has unit-norm rows,
    \item $\bm{C} = \bm{L} \bm{L}^\top$ has unit diagonal: $C_{ii} = \bm{l}_i^\top \bm{l}_i = 1$,
    \item $\bm{C} \succ 0$ because $\bm{L}$ has full rank with positive diagonals.
\end{itemize}
Therefore, the constructed $\bm{C}$ is a valid correlation matrix.

\noindent \emph{Part 2, Any correlation matrix corresponds to a unique angular set:} Let $\bm{C}$ be any full-rank correlation matrix. Then there exists a unique lower triangular matrix $\bm{L}$ with positive diagonal entries such that $\bm{C} = \bm{L} \bm{L}^\top$ (Cholesky decomposition).

\begin{itemize}
    \item Since $C_{11} = 1$, we must have $l_{11} = 1$,
    \item For $m=2$, define $\omega_{21} = \arctan(l_{21}/l_{22}) \in (-\pi/2, \pi/2)$,
    \item For $m \geq 3$, since $\bm{l}_m \in \mathbb{S}_+^{m-1}$, and the mapping $\xi_m : \Theta_m \to \mathbb{S}_+^{m-1}$ is bijective \cite{blumenson1960}, there exists a unique $\Omega_m$ such that $\bm{l}_m = \xi_m(\Omega_m)$.
\end{itemize}
Thus, we can uniquely recover the angular coordinates from any given correlation matrix.

\noindent \emph{Conclusion:} This establishes a one-to-one correspondence between the angular coordinate set $\{\omega_{21}, \Omega_3, \ldots, \Omega_M\}$ and the space of correlation matrices via their Cholesky factors.
\end{proof}

\noindent \textbf{Tuning parameters:} Such a transformation enables the direct application of the GLASD algorithm---developed on hyperrectangular domains---to the space of correlation matrices. The default tuning parameter values are set as follows: initial step sizes \( s_j = 0.1 \); initial directional probabilities \( p_j = \frac{1}{2n} \), where \( n \) denotes the dimensionality of the transformed hyperrectangle; \( s_{\text{inc}} = 2 \), \( s_{\text{dec}} = 2 \), \( p_{\text{inc}} = 2 \), \( p_{\text{dec}} = 2 \); \( m = 5 \); \( c = 0.001 \log(n) \); maximum iterations \( T \approx 3000 \log(n) \); stagnation window \( M = 4n \); and convergence tolerance \( \epsilon = 10^{-20} \). Finally, the exploratory step radius \( r \) is dynamically set to the absolute difference between the current coordinate value and the corresponding upper (or lower) bound, depending on the sign of the proposed movement.

\section{Benchmark Study}\label{sec:GLASD_benchmark}
\noindent To evaluate the performance of our proposed optimizer on structured, high-dimensional constrained domains, we adapt four classical benchmark functions---Ackley, Griewank, Rastrigin, and Rosenbrock---to operate over the space of correlation matrices. Specifically, given a correlation matrix $\boldsymbol{C} \in \mathbb{R}^{M \times M}$, we extract and vectorize its off-diagonal elements into a vector $\boldsymbol{x} \in \mathbb{R}^{d}$, where $d = M(M-1)$. Each function is applied to $\boldsymbol{x}$ after appropriate scaling, in line with standard benchmarking conventions. The modified Ackley function is evaluated on inputs $x_i = 10 \cdot C_{pq}$, $p\neq q \in \{1,\ldots, M\}$, capturing flat regions and oscillatory behavior through exponential and cosine terms. The Griewank function is applied to inputs scaled as $x_i = 100 \cdot C_{pq}$, combining smooth quadratic growth with a product of cosine modulations, challenging optimizers with moderate multimodality. The Rastrigin function, evaluated on $x_i = 10 \cdot C_{pq}$, is known for its highly repetitive landscape with many regularly spaced local minima, testing the algorithm’s robustness in oscillatory terrains. Lastly, the Rosenbrock function, applied to $x_i = 100 \cdot C_{pq}$, introduces a narrow, curved valley leading to the global minimum, requiring careful navigation through non-linear dependencies. These transformations preserve the structural complexity of the original functions while embedding them in the space of positive definite correlation matrices with unit diagonals, making them suitable for stress-testing black-box optimizers under non-Euclidean and constraint-aware settings. Additionally, note that under this reformulation of the objective functions on the space of correlation matrices, the global minimum is attained at the identity matrix. We refer to \cite{surjanovic2013virtual} for the original definitions and considered domains of these benchmark functions.

In Table~\ref{on_PD_benchmark}, we compare the performance of GLASD against a range of positive definite manifold optimization methods available in MATLAB. Specifically, we include classical constrained optimizers such as Interior-Point, Sequential Quadratic Programming (SQP), and Active-Set algorithms implemented in \texttt{fmincon}, as well as manifold-based solvers from the \texttt{Manopt} toolbox~\cite{boumal2014manopt}. The maximum computation time for each method was capped at one hour. For each setting, optimization was repeated 10 times using the same set of randomly chosen initial points across all methods. The table reports the minimum objective value obtained across the 10 runs, the corresponding standard error, and the mean runtime (in seconds, with standard error in parentheses) for correlation matrices of dimensions $M \times M$ with $M = 5, 10, 20, 50$. All computations are performed on a desktop running the Windows 10 Enterprise operating system desktop with 32 GB RAM and the following processor characteristics:  12th Gen Intel(R) Core(TM) i7-12700, 2100 Mhz, 12 Cores(s), 20 Logical Processor(s).

Across all four benchmark functions and increasing matrix dimensions, GLASD consistently delivers competitive or superior performance relative to both \texttt{fmincon} and the \texttt{Manopt} solvers. Notably, for the \textit{Ackley} and \textit{Rastrigin} functions—known for their highly multimodal and oscillatory landscapes—GLASD achieves lower objective values while maintaining computational efficiency. In contrast, traditional solvers such as \texttt{fmincon} frequently converge to suboptimal regions or experience numerical instability, particularly in higher dimensions. On the \textit{Griewank} function, which combines smooth quadratic and periodic components, GLASD exhibits favorable stability, achieving near-optimal solutions with minimal variation, particularly for larger dimensions where classical methods struggle with the cosine product term. For the \textit{Rosenbrock} function, which presents a narrow curved valley around the optimum, GLASD matches or surpasses manifold-based methods in lower to moderate dimensions while maintaining significantly reduced computation time.

Overall, the results validate GLASD's adaptability and robustness across diverse function classes and dimensional scales, without relying on gradient information or function-specific tuning. Its ability to explore highly constrained, non-Euclidean domains such as the correlation matrix space makes it a promising tool for black-box optimization in structured statistical inference tasks.

\begin{table*}[htbp]
\caption{Performance comparison of GLASD and baseline optimization algorithms across four (modified) benchmark objective functions (Ackley, Griewank, Rastrigin, and Rosenbrock), evaluated over increasing correlation matrix dimensions ($M \times M$ for $M = 5$, 10, 20, 50). For each setting, results are based on 10 runs initialized from the same set of randomly chosen starting points across all methods. The reported minimum objective value is the best among the 10 runs, while the mean runtime (in seconds) and standard error (in parenthesis) are computed across those 10 runs.}
\label{on_PD_benchmark}
\centering
\resizebox{1.95\columnwidth}{!}{%
\begin{tabular}{l|l|ccc|ccc|ccc|ccc}
\hline
\multirow{2}{*}{Functions} & \multirow{2}{*}{Methods} & \multicolumn{3}{c|}{$M = 5$} & \multicolumn{3}{c|}{$M = 10$} & \multicolumn{3}{c|}{$M = 20$} & \multicolumn{3}{c}{$M = 50$} \\ \cline{3-14} 
 &  & min. value & s.e. of values & mean time (s.e.) & min. value & s.e. of values & mean time (s.e.) & min. value & s.e. of values & mean time (s.e.) & min. value & s.e. of values & mean time (s.e.) \\ \hline
\multirow{8}{*}{Ackley} & GLASD & \textbf{1.85E-01} & 2.93E-01 & 0.05 (0.016) & \textbf{1.25E-07} & 7.99E-01 & 0.25 (0.069) & \textbf{1.21E-05} & 3.26E-01 & 9.47 (3.598) & \textbf{1.15E+00} & 1.27E-01 & 24.00 (0.660) \\
 & fmincon:active-set & 4.77E+00 & 6.92E-01 & 0.10 (0.038) & 2.71E+00 & 3.26E-01 & 0.25 (0.038) & 2.24E+00 & 1.97E-01 & 1.23 (0.088) & 2.58E+00 & 8.08E-02 & 14.84 (1.938) \\
 & fmincon:interior-point & 7.62E+00 & 5.15E-01 & 0.11 (0.031) & 2.76E+00 & 2.95E-01 & 0.22 (0.033) & 3.65E+00 & 1.25E-01 & 0.30 (0.031) & 4.12E+00 & 1.75E-01 & 0.74 (0.023) \\
 & fmincon:sqp & 2.32E+00 & 6.65E-01 & 0.03 (0.005) & \textbf{1.57E-01} & 3.47E-01 & 0.21 (0.008) & 2.11E+00 & 1.88E-01 & 3.41 (0.194) & 2.75E+00 & 6.34E-02 & 36.28 (7.507) \\
 & Manopt:barzilai-borwein & \textbf{2.59E-04} & 1.56E+00 & 1.20 (0.160) & 2.19E+00 & 2.09E-01 & 12.90 (0.106) & 2.74E+00 & 1.91E-01 & 1845.61 (16.090) & 2.91E+00 & 4.25E-02 & 3655.08 (9.139) \\
 & Manopt:conjugate-gradient & 3.40E+00 & 9.46E-01 & 0.34 (0.141) & 5.33E+00 & 4.60E-01 & 10.06 (1.096) & 2.59E+00 & 3.14E-01 & 1473.70 (35.334) & 2.75E+00 & 2.94E-02 & 3628.97 (8.379) \\
 & Manopt:steepest-descent & 2.01E+00 & 1.11E+00 & 0.90 (0.188) & 2.25E+00 & 5.86E-01 & 12.29 (0.647) & \textbf{1.37E+00} & 9.70E-02 & 1562.16 (84.209) & \textbf{3.92E-02} & 1.94E-01 & 3652.13 (12.547) \\
 & Manopt:trust-region & 8.58E+00 & 5.89E-01 & 3.41 (0.093) & 5.27E+00 & 4.55E-01 & 125.19 (57.512) & 2.13E+00 & 2.51E-01 & 3612.68 (5.701) & 2.12E+00 & 7.16E-02 & 3739.28 (22.810) \\ \hline
\multirow{8}{*}{Griewank} & GLASD & 7.53E-05 & 3.51E-01 & 0.02 (0.002) & 1.32E-02 & 1.29E-01 & 0.18 (0.032) & 2.00E-05 & 5.02E-01 & 3.15 (0.486) & 6.89E+00 & 4.95E-01 & 24.19 (0.454) \\
 & fmincon:active-set & 3.54E-08 & 1.30E-01 & 0.06 (0.014) & 9.67E-01 & 1.50E-02 & 0.49 (0.048) & 1.12E+00 & 4.90E-02 & 8.87 (0.684) & 1.05E+00 & 8.26E-02 & 711.20 (49.887) \\
 & fmincon:interior-point & \textbf{3.62E-14} & 1.34E-01 & 0.09 (0.008) & 5.31E-09 & 1.14E-01 & 0.17 (0.003) & 8.66E-02 & 1.00E-01 & 0.26 (0.007) & 1.58E+01 & 6.60E-01 & 0.75 (0.012) \\
 & fmincon:sqp & 5.24E-12 & 1.35E-01 & 0.02 (0.001) & 2.06E-04 & 1.03E-01 & 0.16 (0.018) & 7.09E-03 & 9.69E-02 & 5.70 (2.323) & \textbf{8.09E-02} & 7.88E-02 & 523.22 (129.671) \\
 & Manopt:barzilai-borwein & 1.19E-01 & 1.22E-01 & 1.30 (0.155) & 1.22E-11 & 1.29E-01 & 12.23 (0.828) & 3.99E-10 & 9.02E-02 & 1645.40 (33.268) & 1.18E+00 & 1.63E-01 & 3651.76 (12.201) \\
 & Manopt:conjugate-gradient & 3.02E-13 & 1.09E-01 & 0.20 (0.079) & \textbf{1.51E-12} & 1.13E-01 & 7.40 (1.619) & \textbf{4.14E-12} & 7.48E-07 & 1006.67 (119.258) & \textbf{2.80E-02} & 9.04E-02 & 3648.72 (8.335) \\
 & Manopt:steepest-descent & 3.65E-13 & 1.48E-01 & 1.03 (0.205) & 3.20E-12 & 1.03E-01 & 11.63 (1.040) & 1.55E-04 & 7.65E-02 & 1412.82 (8.442) & 6.05E-01 & 5.84E-02 & 3666.81 (8.680) \\
 & Manopt:trust-region & \textbf{6.12E-14} & 1.40E-01 & 2.95 (0.459) & \textbf{8.57E-13} & 1.03E-01 & 88.85 (33.392) & \textbf{6.85E-11} & 5.54E-02 & 3424.90 (183.014) & 4.03E+00 & 3.86E-01 & 3977.01 (71.034) \\ \hline
\multirow{8}{*}{Rastrigin} & GLASD & \textbf{1.01E+01} & 2.11E+01 & 0.03 (0.011) & \textbf{1.13E+02} & 5.46E+01 & 0.36 (0.079) & \textbf{6.83E+02} & 7.48E+01 & 4.20 (0.609) & \textbf{3.40E+03} & 1.92E+02 & 22.22 (0.702) \\
 & fmincon:active-set & 5.97E+01 & 5.47E+01 & 0.12 (0.050) & 5.69E+02 & 4.79E+01 & 2.83 (0.112) & NaN & NaN & NaN (NaN) & NaN & NaN & NaN (NaN) \\
 & fmincon:interior-point & 1.53E+02 & 4.04E+01 & 0.15 (0.040) & 3.78E+02 & 6.08E+01 & 0.33 (0.043) & 2.11E+03 & 1.25E+02 & 0.37 (0.022) & 1.95E+04 & 7.45E+02 & 0.97 (0.045) \\
 & fmincon:sqp & 4.78E+01 & 2.55E+01 & 0.04 (0.008) & 3.86E+02 & 2.92E+01 & 0.83 (0.024) & NaN & NaN & NaN (NaN) & NaN & NaN & NaN (NaN) \\
 & Manopt:barzilai-borwein & \textbf{1.99E+00} & 1.54E+01 & 2.66 (0.106) & \textbf{2.39E+01} & 1.09E+01 & 19.19 (0.098) & \textbf{2.24E+02} & 2.26E+02 & 1232.87 (111.907) & 2.06E+04 & 8.76E+02 & 3637.29 (9.777) \\
 & Manopt:conjugate-gradient & 2.37E+02 & 5.28E+01 & 0.51 (0.193) & 7.32E+02 & 2.91E+01 & 15.83 (1.656) & 1.67E+03 & 3.57E+01 & 1320.30 (14.472) & \textbf{6.34E+03} & 9.65E+01 & 3648.68 (9.919) \\
 & Manopt:steepest-descent & 1.93E+02 & 4.66E+01 & 1.69 (0.308) & 6.27E+02 & 4.69E+01 & 19.08 (0.043) & 1.59E+03 & 4.92E+01 & 1341.30 (16.733) & 8.07E+03 & 2.55E+02 & 3641.47 (9.235) \\
 & Manopt:trust-region & 2.37E+02 & 5.38E+01 & 5.55 (0.155) & 7.32E+02 & 2.86E+01 & 129.41 (60.144) & 1.61E+03 & 2.70E+01 & 3615.28 (6.706) & 8.33E+03 & 1.93E+02 & 3881.41 (51.192) \\ \hline
\multirow{8}{*}{Rosenbrock} & GLASD & \textbf{5.15E-11} & 1.40E+08 & 0.05 (0.009) & \textbf{8.81E+01} & 2.50E-02 & 0.42 (0.052) & \textbf{3.76E+02} & 4.96E+07 & 3.89 (1.022) & 2.44E+08 & 1.60E+07 & 21.47 (0.249) \\
 & fmincon:active-set & 6.73E+03 & 2.45E+05 & 0.09 (0.008) & 1.17E+04 & 9.86E+03 & 2.01 (0.041) & 9.96E+04 & 2.02E+04 & 377.28 (5.830) & NaN & NaN & NaN (NaN) \\
 & fmincon:interior-point & 2.68E-08 & 1.88E+00 & 0.11 (0.020) & 7.63E+02 & 7.18E+04 & 0.14 (0.002) & 4.13E+08 & 2.30E+08 & 0.20 (0.004) & 1.01E+09 & 2.36E+09 & 0.76 (0.030) \\
 & fmincon:sqp & 1.52E+02 & 2.43E+06 & 0.03 (0.002) & \textbf{1.61E+02} & 6.12E+05 & 0.87 (0.176) & 2.03E+03 & 4.59E+05 & 202.14 (18.788) & NaN & NaN & NaN (NaN) \\
 & Manopt:barzilai-borwein & 4.05E+09 & 4.25E+09 & 0.01 (0.001) & 1.63E+10 & 1.76E+09 & 0.02 (0.001) & 1.98E+10 & 1.56E+09 & 1.60 (0.044) & 2.48E+10 & 7.10E+08 & 119.75 (2.236) \\
 & Manopt:conjugate-gradient & 1.88E+01 & 9.28E-10 & 0.15 (0.017) & \textbf{8.81E+01} & 3.71E-02 & 8.03 (1.743) & \textbf{3.75E+02} & 4.29E-02 & 901.03 (174.999) & \textbf{8.06E+03} & 2.64E+03 & 3654.56 (8.325) \\
 & Manopt:steepest-descent & 1.88E+01 & 7.41E-10 & 0.38 (0.128) & \textbf{8.81E+01} & 5.23E+00 & 10.31 (1.460) & \textbf{3.75E+02} & 2.71E+00 & 958.70 (94.386) & \textbf{6.34E+03} & 1.69E+03 & 3643.65 (8.871) \\
 & Manopt:trust-region & \textbf{1.34E-09} & 1.88E+00 & 3.25 (0.019) & \textbf{8.81E+01} & 7.86E-04 & 66.38 (18.825) & \textbf{3.75E+02} & 1.38E-01 & 3255.01 (105.917) & 4.31E+07 & 2.78E+07 & 3926.85 (55.487) \\ \hline
\end{tabular}}
\end{table*}

\section{Simulation Study}\label{sec:simulation_study}
\noindent We conducted an extensive simulation study to evaluate the robustness and accuracy of correlation matrix estimators under various contamination settings and distributional assumptions. In addition to the classical estimator based on the Gaussian log-likelihood, we assessed three robust alternatives, each defined by a different loss function: the Huber loss, the truncated Mahalanobis loss, and Tukey’s biweight loss, as introduced in earlier sections. For all three loss functions, we employed an interquartile range (IQR)-based cutoff for outlier detection—specifically, the threshold parameter (i.e., $\delta$ in \eqref{eq:huber_likelihood}, and $\tau$ in \eqref{eq:truncated_loss} and \eqref{eq:tukey_likelihood}) was set to the third quartile plus three times the IQR.

Four types of contamination scenarios were considered:

\begin{enumerate}
    \item \emph{Row contamination}, where 10\% of the rows were randomly selected and 30--70\% of their entries were perturbed by adding $10$;
    \item \emph{Column contamination}, where 10\% of the columns were selected, with 30--70\% of their entries similarly shifted;
    \item \emph{Random contamination}, where 5\% of the elements in the data matrix were randomly chosen and each was increased by $100$;
    \item \emph{Heavy-tailed distribution}, where samples were drawn from a multivariate $t$-distribution with 3 degrees of freedom.
\end{enumerate}

Each contamination type was examined under three different true correlation structures:
\begin{itemize}
    \item A dense \emph{random (non-sparse)} correlation matrix with all non-zero entries,
    \item A \emph{sparse} matrix with 90\% zero off-diagonal entries and remaining non-zero entries drawn uniformly from $[0.1, 0.3]$,
    \item A \emph{block-Toeplitz} structure with three variable blocks (25\%, 50\%, 25\%), each exhibiting Toeplitz structure with correlation decay parameters 0.6, 0.3, and 0.4, respectively.
\end{itemize}

Simulations were performed for two dimensional settings: $(p, n) = (20, 100)$ and $(50, 500)$. For each combination of distributional scenario and correlation structure, we generated 10 independent datasets. For each dataset, the robust estimators were obtained by applying GLASD optimization separately for each loss function, initialized from 10 distinct random starting points. The solution with the minimum objective value across these runs was retained as the final estimator.

Performance was measured using the root mean squared error (RMSE) between the estimated and true correlation matrices. The reported RMSE values in Table~\ref{table:RMSE} represent the average over the 10 independently generated datasets for each setting, with standard errors shown in parentheses. The best-performing method in each setting is highlighted in bold.

Across all simulation settings summarized in Table~\ref{table:RMSE}, we observe clear advantages of robust correlation matrix estimators over the classical Gaussian likelihood-based estimator in the presence of data contamination and heavy-tailed distributions. While the Gaussian estimator performs relatively well under row and column contamination when the true correlation matrix is dense and non-sparse, its accuracy deteriorates notably in heavy-tailed settings, particularly when samples are drawn from a multivariate $t$-distribution with low degrees of freedom. The degradation is especially evident in sparse or structured (e.g., block-Toeplitz) correlation matrices, where contamination interacts more adversely with limited signal strength.

The Huber loss-based estimator consistently achieves the lowest RMSE across a wide range of contamination scenarios, demonstrating a favorable balance between robustness and statistical efficiency. It performs particularly well under row contamination and for sparse or block-structured matrices, indicating resilience to structured outliers. The estimator based on truncated Mahalanobis loss shows strong performance in scenarios involving column contamination and $t$-distributed data, outperforming other methods in these settings due to its ability to effectively limit the influence of extreme residuals. Tukey’s biweight loss, while theoretically robust, tends to yield less stable estimates in higher-dimensional settings, and its performance lags behind that of Huber and truncated losses in most scenarios. 

Overall, these results highlight the utility of combining robust loss functions with globally convergent optimization strategies such as GLASD. Robust estimators offer a practical and reliable solution when modeling correlation structures in the presence of contamination or distributional misspecification, ensuring improved accuracy across diverse structural and distributional regimes.

\begin{table}[t]
\centering
\caption{Summary of Root Mean Squared Error (RMSE) Estimates with Standard Errors for Correlation Matrix Estimators Across Contaminated Gaussian and t-distributions Under Varying Sample Sizes and true correlation matrix structures.}
  \label{table:RMSE}
\renewcommand{\arraystretch}{1.3}
\resizebox{.9\columnwidth}{!}{%
\begin{tabular}{c|c|c|cccc}
\hline
Distribution & \begin{tabular}[c]{@{}c@{}}Corr. matrix\\ type\end{tabular} & $(p,n) $ & Gaussian & Huber & Truncated & Tukey \\ \hline
\multirow{6}{*}{\begin{tabular}[c]{@{}c@{}}Gaussian,\\ Row \\ Contamination\end{tabular}} & \multirow{2}{*}{\begin{tabular}[c]{@{}c@{}}Random\\ (non-sparse)\end{tabular}} & 20, 100 & 0.220 (0.0044) & \textbf{0.207 (0.0030)} & 0.240 (0.0043) & 0.270 (0.0037) \\
 &  & 50, 500 & \textbf{0.146 (0.0010)} & 0.163 (0.0011) & 0.175 (0.0012) & 0.186 (0.0013) \\ \cline{2-7} 
 & \multirow{2}{*}{\begin{tabular}[c]{@{}c@{}}Sparse \\ Uniform\end{tabular}} & 20, 100 & 0.348 (0.0094) & \textbf{0.178 (0.0039)} & 0.229 (0.0048) & 0.231 (0.0041) \\
 &  & 50, 500 & 0.160 (0.0020) & \textbf{0.124 (0.0017)} & 0.146 (0.0014) & 0.145 (0.0012) \\ \cline{2-7} 
 & \multirow{2}{*}{\begin{tabular}[c]{@{}c@{}}Block\\ Toeplitz\end{tabular}} & 20, 100 & 0.351 (0.0097) & \textbf{0.176 (0.0044)} & 0.204 (0.0066) & 0.233 (0.0059) \\
 &  & 50, 500 & 0.136 (0.0016) & \textbf{0.111 (0.0020)} & 0.137 (0.0023) & 0.152 (0.0022) \\ \hline
\multirow{6}{*}{\begin{tabular}[c]{@{}c@{}}Gaussian,\\ Column \\ Contamination\end{tabular}} & \multirow{2}{*}{\begin{tabular}[c]{@{}c@{}}Random\\ (non-sparse)\end{tabular}} & 20, 100 & \textbf{0.226 (0.0062)} & 0.236 (0.0089) & 0.237 (0.0091) & 0.247 (0.0043) \\
 &  & 50, 500 & \textbf{0.142 (0.0016)} & \textbf{0.142 (0.0019)} & 0.143 (0.0018) & 0.174 (0.0010) \\ \cline{2-7} 
 & \multirow{2}{*}{\begin{tabular}[c]{@{}c@{}}Sparse \\ Uniform\end{tabular}} & 20, 100 & 0.152 (0.0040) & 0.153 (0.0040) & \textbf{0.151 (0.0045)} & 0.154 (0.0039) \\
 &  & 50, 500 & 0.097 (0.0010) & 0.097 (0.0017) & \textbf{0.093 (0.0008)} & 0.115 (0.0012) \\ \cline{2-7} 
 & \multirow{2}{*}{\begin{tabular}[c]{@{}c@{}}Block\\ Toeplitz\end{tabular}} & 20, 100 & \textbf{0.113 (0.0032)} & \textbf{0.113 (0.0040)} & \textbf{0.113 (0.0031)} & 0.159 (0.0048) \\
 &  & 50, 500 & 0.060 (0.0009) & 0.060 (0.0012) & \textbf{0.057 (0.0009)} & 0.119 (0.0024) \\ \hline
\multirow{6}{*}{\begin{tabular}[c]{@{}c@{}}Gaussian,\\ Random \\ Contamination\end{tabular}} & \multirow{2}{*}{\begin{tabular}[c]{@{}c@{}}Random\\ (non-sparse)\end{tabular}} & 20, 100 & \textbf{0.132 (0.0047)} & 0.133 (0.0046) & 0.137 (0.0076) & 0.213 (0.0054) \\
 &  & 50, 500 & 0.108 (0.0008) & \textbf{0.104 (0.0008)} & 0.106 (0.0007) & 0.153 (0.0011) \\ \cline{2-7} 
 & \multirow{2}{*}{\begin{tabular}[c]{@{}c@{}}Sparse \\ Uniform\end{tabular}} & 20, 100 & \textbf{0.132 (0.0026)} & \textbf{0.132 (0.0025)} & 0.143 (0.0075) & 0.196 (0.0054) \\
 &  & 50, 500 & 0.051 (0.0006) & \textbf{0.050 (0.0007)} & \textbf{0.050 (0.0003)} & 0.091 (0.0015) \\ \cline{2-7} 
 & \multirow{2}{*}{\begin{tabular}[c]{@{}c@{}}Block\\ Toeplitz\end{tabular}} & 20, 100 & 0.101 (0.0035) & \textbf{0.098 (0.0030)} & 0.112 (0.0044) & 0.152 (0.0047) \\
 &  & 50, 500 & \textbf{0.042 (0.0009)} & 0.043 (0.0014) & 0.043 (0.0007) & 0.081 (0.0022) \\ \hline
\multirow{6}{*}{\begin{tabular}[c]{@{}c@{}}t-distribution\\ (df = 3)\end{tabular}} & \multirow{2}{*}{\begin{tabular}[c]{@{}c@{}}Random\\ (non-sparse)\end{tabular}} & 20, 100 & 0.184 (0.0149) & \textbf{0.162 (0.0082)} & 0.175 (0.0105) & 0.231 (0.0057) \\
 &  & 50, 500 & 0.117 (0.0038) & \textbf{0.115 (0.0029)} & 0.140 (0.0013) & 0.169 (0.0017) \\ \cline{2-7} 
 & \multirow{2}{*}{\begin{tabular}[c]{@{}c@{}}Sparse \\ Uniform\end{tabular}} & 20, 100 & 0.213 (0.0246) & 0.183 (0.0151) & \textbf{0.135 (0.0101)} & 0.172 (0.0065) \\
 &  & 50, 500 & 0.101 (0.0050) & \textbf{0.085 (0.0018)} & 0.089 (0.0026) & 0.117 (0.0023) \\ \cline{2-7} 
 & \multirow{2}{*}{\begin{tabular}[c]{@{}c@{}}Block\\ Toeplitz\end{tabular}} & 20, 100 & 0.217 (0.0245) & 0.166 (0.0116) & \textbf{0.129 (0.0073)} & 0.173 (0.0068) \\
 &  & 50, 500 & 0.069 (0.0052) & \textbf{0.054 (0.0016)} & 0.073 (0.0034) & 0.118 (0.0038) \\ \hline
\end{tabular}}
\end{table}

\section{Case Study: Robust Correlation Estimation from Proteomic Data}\label{sec:TCPA_analysis}
\noindent We illustrate the practical utility of our method using proteomic measurements from the TCPA breast cancer dataset~\cite{li2013tcpa}, which includes abundance levels for hundreds of proteins across patient tumor samples. For this analysis, we focus on a curated subset of $p = 20$ proteins implicated in breast cancer~\cite{das2020nexus}, spanning four major biological pathways: \emph{breast reactive}, \emph{cell cycle}, \emph{hormone receptor}, and \emph{hormone signaling}. The dataset comprises $n = 879$ patient samples.

To assess potential data contamination, we first evaluate outliers using an interquartile range (IQR)-based rule applied independently to each protein. For each variable, observations falling beyond $[Q_1 - 1.5 \cdot \text{IQR}, Q_3 + 1.5 \cdot \text{IQR}]$ are flagged as outliers. The total outlier count per protein is summarized in Fig.~\ref{fig:casestudy_plot}(a), where bar colors indicate pathway membership. Several proteins, particularly from the \emph{cell cycle} and \emph{breast reactive} pathways, exhibit a large number of outlying measurements, underscoring the need for robust estimation strategies.

\begin{figure}[h]
  \centering
  \includegraphics[width=0.48\textwidth]{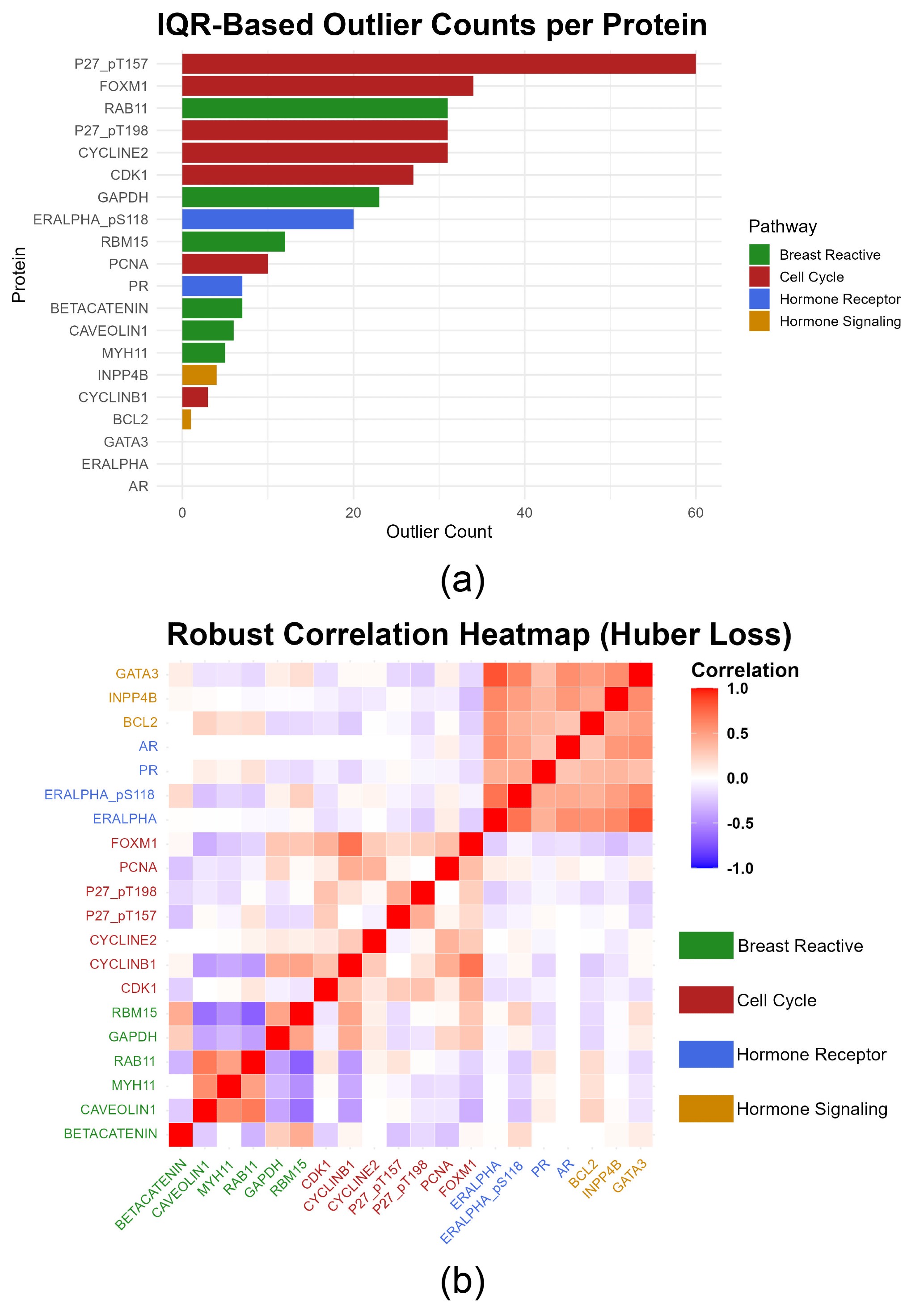}
   \caption{Figure: Protein-level expression characteristics in four major activated pathways for breast cancer. \textbf{(a)} IQR-based outlier counts across 20 proteins, colored by their associated biological pathways. \textbf{(b)} Heatmap of pairwise Pearson correlations between the same proteins, estimated robustly using Huber loss. Protein axis labels are colored according to pathway grouping. Distinct correlation blocks can be observed within pathway groups, particularly among cell cycle and hormone receptor proteins.}
  \label{fig:casestudy_plot}
\end{figure}

Using the proposed GLASD algorithm with the Huber loss function, we estimate the correlation matrix of the 20 proteins, as shown in Fig.~\ref{fig:casestudy_plot}(b). The heatmap reveals coherent within-pathway associations, with proteins belonging to the same biological module—particularly \emph{cell cycle}, \emph{hormone receptor}, and \emph{hormone signaling} pathways—exhibiting relatively stronger mutual correlations. These patterns give rise to a block-diagonal structure, suggesting modular organization among pathway-specific components. In contrast, inter-pathway correlations tend to be weaker and more heterogeneous, reflecting limited co-regulation across distinct functional axes. The \emph{breast reactive} proteins exhibit comparatively less internal coherence, potentially due to pathway heterogeneity or measurement variability. While no direct comparison to non-robust estimates is shown, the application of a robust loss function provides a principled approach to limit the impact of extreme observations on inferred correlation structure, enabling more stable and interpretable network summaries in the presence of outliers.

\section{Conclusion}\label{sec:conclusion}
In this work, we developed a novel optimization algorithm, GLASD (\textit{Global Adaptive Stochastic Descent}), designed to perform robust correlation matrix estimation under non-convex loss functions and positive definiteness constraints. By leveraging a fully gradient-free, globally explorative, and coordinate-adaptive search scheme, GLASD overcomes the limitations of classical constrained optimization techniques and local descent-based solvers, particularly in high-dimensional and contaminated data settings.

We demonstrated the efficacy of GLASD across a range of synthetic benchmarks and real-world proteomics applications. Specifically, we applied GLASD to minimize robust loss functions such as Huber, truncated Mahalanobis, and Tukey's biweight, showcasing improved resilience to outliers and heavy-tailed noise compared to the Gaussian log-likelihood approach. A key enabler of this flexibility was a novel transformation that reparameterizes the space of correlation matrices onto a hyper-rectangle, making it compatible with general-purpose black-box optimization algorithms.

Theoretical guarantees on global convergence were established under mild conditions, and further insights into the deterministic variant of the algorithm (ASD) were provided under convexity assumptions. Empirically, GLASD consistently outperformed or matched classical solvers and manifold-based methods on stress-test functions defined over the correlation matrix manifold, highlighting its versatility and robustness.

Taken together, GLASD provides a powerful and general-purpose framework for constrained optimization on correlation matrices, particularly useful in settings where the objective is non-convex, non-differentiable, or defined only via simulations. Future work includes extending GLASD to other matrix manifolds (e.g., precision matrices or low-rank structures), exploring adaptive cooling schedules for faster convergence, and developing scalable implementations for large-scale biological network estimation.

\section*{Software}
\noindent The GLASD optimization algorithm, implemented in MATLAB, is available on GitHub at \href{https://github.com/priyamdas2/GLASD}{https://github.com/priyamdas2/GLASD}.

\section*{Acknowledgments}
PD was partially supported by NIH/NCI Cancer Center Support Grant P30 CA016059.

\vspace{1pt}
\begin{IEEEbiography}[{\includegraphics[width=1in,height=1.2in,clip,keepaspectratio]{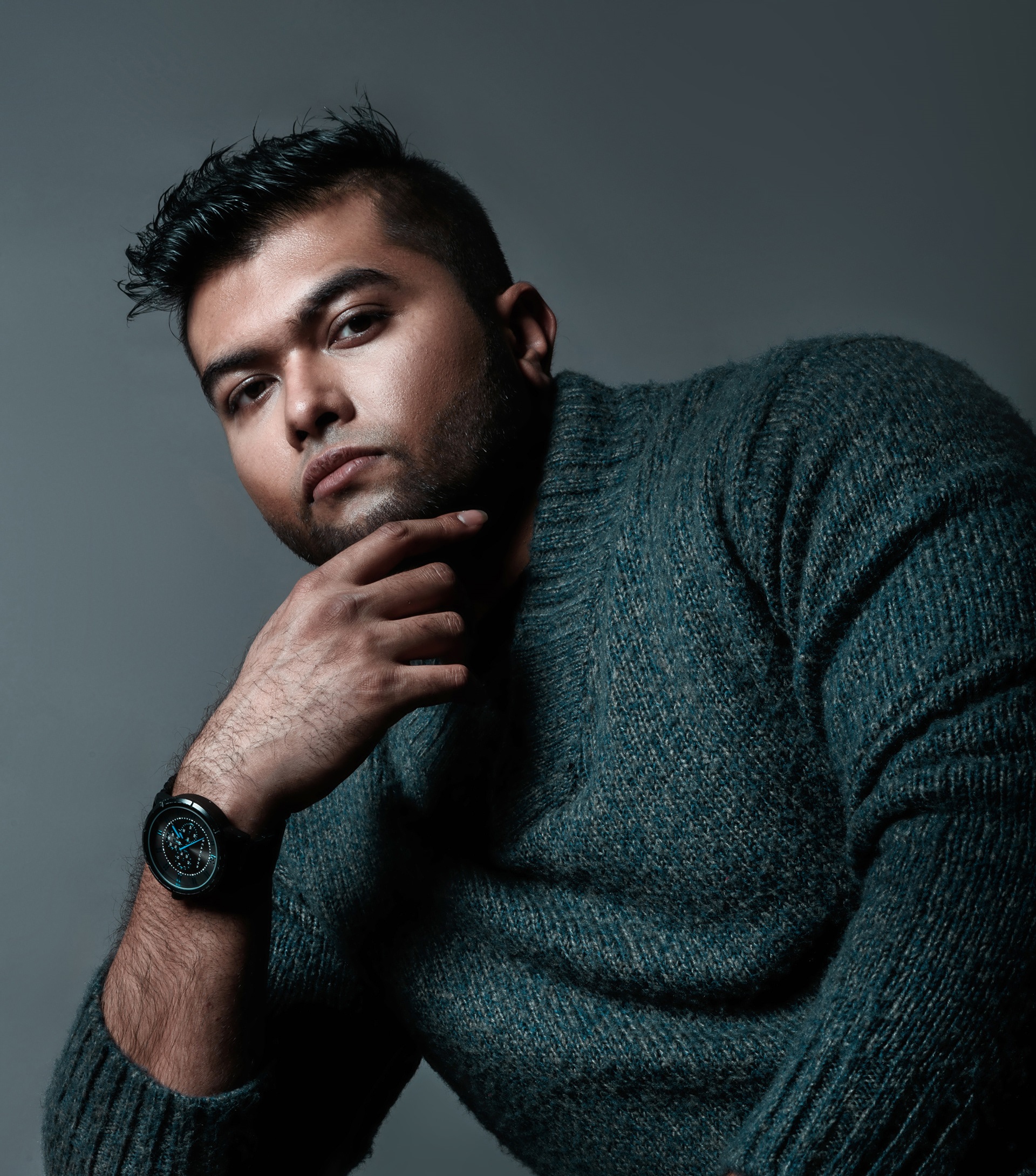}}]{Priyam Das} received the B.Stat. and M.Stat. degrees in statistics from the Indian Statistical Institute, Kolkata, India, in 2011 and 2013, respectively, and the Ph.D. degree in statistics from North Carolina State University, Raleigh, NC, USA, in 2016.

Dr. Das is currently an Assistant Professor in the Department of Biostatistics at Virginia Commonwealth University and a member of the Massey Cancer Center. His research interests include statistical machine learning, black-box optimization, graphical models, high-dimensional inference, and applications in cancer genomics and biomedical network analysis.
\end{IEEEbiography}



\vfill

\end{document}